%% file: arXiv.tex
\documentclass{article}
\input{00.pakages.tex}

\input{00.commands.tex}

\title{Optimally Interpolating between Ex-Ante Fairness and Welfare}

\author{Mikael Høgsgaard \and Panagiotis Karras \and Wenyue Ma \and  Nidhi Rathi \and Chris Schwiegelshohn} \date{Department of Computer Science, Aarhus University\\ {\AA}bogade 34, 8200 Aarhus N, Denmark\\ Email: \url{{hogsgaard, panos, wenyuema, nidhi, schwiegelshohn }@cs.au.dk}}

\begin{document}

\maketitle

\begin{abstract}
For the fundamental problem of allocating a set of resources among individuals with varied preferences, the quality of an allocation relates to the degree of fairness and the collective welfare achieved. Unfortunately, in many resource-allocation settings, it is computationally hard to maximize welfare while achieving fairness goals. 

In this work, we consider ex-ante notions of fairness; popular examples include the \emph{randomized round-robin algorithm} and \emph{sortition mechanism}. We propose a general framework to systematically study the \emph{interpolation} between fairness and welfare goals in a multi-criteria setting. We develop two efficient algorithms ($\emix$ and $\smix$) that achieve different trade-off guarantees with respect to fairness and welfare. $\emix$ achieves an optimal multi-criteria approximation with respect to fairness and welfare, while $\smix$ achieves optimality up to a constant factor with zero computational overhead beyond the underlying \emph{welfare-maximizing mechanism} and the \emph{ex-ante fair mechanism}. Our framework makes no assumptions on either of the two underlying mechanisms, other than that the fair mechanism produces a distribution over the set of all allocations. Indeed, if these mechanisms are themselves approximation algorithms, our framework will retain the approximation factor, guaranteeing sensitivity to the quality of the underlying mechanisms, while being \emph{oblivious} to them. 
We also give an extensive experimental analysis for the aforementioned ex-ante fair mechanisms on real data sets, confirming our theoretical analysis.
\end{abstract}

\input{sections/01.introduction}

\input{sections/02.notation.tex}

\input{sections/03-algorithm-arxiv.tex}
\input{sections/04-experiments-arxiv.tex}

\bibliography{ref}

\end{document}

%% file: 00.pakages.tex
\usepackage[utf8]{inputenc}
\usepackage[margin=1in]{geometry}

\usepackage{times}
\usepackage{soul}
\usepackage{url}
\usepackage[hidelinks]{hyperref}
\usepackage{float}

\usepackage{tocloft}
\usepackage[small]{caption}
\usepackage{booktabs}
\usepackage[switch]{lineno}
\usepackage{pgfplots}
\usepackage{graphicx}
\usepackage{tikz}
\usetikzlibrary{patterns}

\usepackage{mdframed}
\usepackage{amsmath,amsthm,amssymb}
\usepackage{xspace}
\usepackage{url}

\usepackage{algorithmic}
\usepackage{cleveref}

\usepackage{appendix}
\usepackage{babel}
\usepackage{bbm}
\usepackage{slashed}

\usepackage{multirow}
\usepackage{tablefootnote}
\usepackage{afterpage}
\usepackage{lipsum}
\usepackage{booktabs}
\usepackage{thm-restate}
\usepackage{graphicx}
\usepackage{subcaption}
\usepackage{lipsum}
\usepackage[ruled]{algorithm2e}

\usepackage{array}
\usepackage{xfrac}
\pgfplotsset{compat=1.17}

%% file: 00.commands.tex
\newcommand{\eps}{\varepsilon}
\newcommand{\e}{\mathbb{E}}
\newcommand{\p}{\mathbb{P}}
\newcommand{\pf}{p^{ f}}

\newcommand{\px}{p^{\eps}}
\newcommand{\pa}{p^{s}}
\newcommand{\po}{p^{\textsc{Opt}}}

\newcommand{\ou}{\textsc{Opt}}

\newcommand{\TV}{\mathrm{TV}}
\newcommand{\emix}{\varepsilon\textsc{-Mix}}
\newcommand{\Sol}{\mathcal{S}}
\newcommand{\smix}{\textsc{Simple-Mix}}
\newcommand{\ta}{T^\beta}
\newcommand{\taa}{T^\alpha}
\newcommand{\tz}{T^0}
\newcommand{\wa}{w^{\alpha}}

\newcommand{\paa}{p^{\alpha}}
\newcommand{\paac}{\tilde{p}^{\alpha}}
\newcommand{\mcl}{\mathcal{M}_\lambda}
\newcommand{\rrq}{\textsc{RandomReplace}}
\newcommand{\km}{\textsc{Kmeans++}}
\newcommand{\mmb}{\textsc{Max-Matching}}
\newcommand{\grd}{\textsc{Greedy}}

\newcommand{\instance}{\mathcal{I}(V,\pf, \mathcal{M}_\lambda, \alpha)}

\theoremstyle{plain}
\newtheorem{theorem}{Theorem}[section]
\newtheorem{lemma}[theorem]{Lemma}
\newtheorem{claim}[theorem]{Claim}

\theoremstyle{definition}

\theoremstyle{remark}
\newtheorem{remark}[theorem]{Remark}

\DeclareMathOperator*{\argmax}{arg\,max}

\hypersetup{colorlinks=true,allcolors=blue}

\definecolor{RePu01}{RGB}{252, 197, 192}
\definecolor{RePu02}{RGB}{250, 159, 181}
\definecolor{RePu03}{RGB}{247, 104, 161}
\definecolor{RePu04}{RGB}{221, 52, 151}
\definecolor{RePu05}{RGB}{174, 1, 126}
\definecolor{RePu06}{RGB}{73, 0, 106}

\DeclareTextFontCommand{\fet}{\fontfamily{lmss}\selectfont} 
\DeclareTextFontCommand{\ralg}{\fontfamily{cmtt}\selectfont} 

%% file: sections/01.introduction.tex
\section{Introduction}
Since the dawn of civilization, human beings have been living in a social construct that necessarily requires making group decisions based on possibly varied individual preferences. For problems like allocating limited natural resources, managing airport traffic, matching markets, and assigning courses to students, it is desirable to have fairness on an individual level while achieving global welfare guarantees. The theory of fair division focuses on the fundamental problem of allocating a set of resources among a set of participating individuals with distinct preferences. Such resource-allocation settings have spawned a flourishing line of research over multiple decades in economics, mathematics, and computer science; see \cite{brams1996fair,robertson1998cake,procaccia2015cake} for excellent expositions. 

Fairness and economic efficiency are two pivotal goals for various allocation problems, and therefore several important notions of fairness and economic welfare have been studied extensively in the literature \cite{bouveret2016characterizing,budish2011combinatorial}. Nash social welfare and egalitarian social welfare are considered to capture important measures of economic efficiency \cite{moulin2004fair,caragiannis2019unreasonable}.  The history of fair division has seen a lot of impressive existential as well as hardness results.  That is, we have existential results for meaningfully fair allocations of resources in quite general settings \cite{stromquist1980cut}, but on the other hand, we don’t have their algorithmic counterparts. In many cases, there are indeed hardness or impossibility results that rule out polynomial-time algorithms for finding such allocations \cite{stromquist2008envy}. And not surprisingly, for various allocation settings, the problem of simultaneously achieving fairness and maximizing welfare is difficult as well \cite{nguyen2013survey,shavell2002fairness,roos2010complexity,arunachaleswaran2019fair}.

Consider the setting of \emph{fairly} allocating indivisible items where each item needs to be allocated to a single agent. In a scenario where a single valuable item  needs to be allocated among multiple agents, it is impossible to achieve the classical fairness notion of \emph{envy-freeness}\footnote{An allocation of resources is envy-free if every agent prefers her own share over that of any other agent's share.}. This non-existence has led to significant research on fairness in deterministic allocations of indivisible goods in economics and computer science \cite{bouveret2016fair}, a vast majority of which focuses on relaxed fairness and welfare properties (see the recent survey by \cite{survey2022}). Randomization plays an important role in achieving ex-ante fairness, explored in recent works \cite{aleksandrov2015online,freeman2020best,aziz2020simultaneously,caragiannis2021interim}. In the above example with one item, assigning it uniformly at random to an agent, we achieve \emph{ex-ante envy-freeness} \cite{freeman2020best}, i.e. prior to sampling an agent, no agent has a reason to envy the others as each has an equal chance of receiving the item. 
Hence, studying ex-ante notions of fairness has become a significant line of research. Popular mechanisms like randomized round-robin\footnote{In a randomized round-robin mechanism, agents come in a uniformly random order and act greedily by choosing their most preferred item from the pool of remaining items.} and sortition-based algorithms\footnote{The notion of \emph{sortition} entails a random selection of representatives from the population providing an alternative method to democracy. Here, the goal is to select a panel that is representative of the population such that individuals would ideally be selected to serve on this panel with equal probability.} \cite{freeman2020best,ebadiansortition,flanigan2021fair} explore ex-ante fairness notions in varied settings.  Motivated by the variety, we assume that a \emph{fair mechanism} can be expressed as a random variable over the space of all solutions (i.e., allocations) and samples can be drawn from it. We will refer to the distribution defined by the fair mechanism as a \emph{fair prior}. We measure the welfare of a randomized allocation as the expected welfare it achieves. Given a \emph{fair mechanism} on one hand and a \emph{welfare-maximizing mechanism} on the other, a natural question that arises is how to compute allocations that retain \emph{closeness} to both fairness and welfare guarantees.

In this paper, we examine whether there exists a natural way to formulate the trade-off between the two extremes of fairness and welfare in resource-allocation settings. That is, given the characterizations of two mechanisms, one that defines the fair prior and the other that is welfare-maximizing, is there a computationally efficient way to meaningfully \emph{interpolate} between the two solutions of fairness and welfare? 

 \noindent
\subsection*{Our Contribution:}
We propose a general framework to systematically study the interpolation between the two pivotal goals of achieving fairness and maximizing welfare for allocation problems. We define an instance of \emph{fair-to-welfare interpolation} (FWI) to consist of a fair prior distribution over allocations and a welfare-maximizing mechanism for a given allocation problem. 
The goal is to find a new mechanism that (i) is close to the fair prior with respect to some distance function while (ii) maximizing welfare.

For any desired choice of $\alpha\in (0,1)$, we consider an algorithm to be $\alpha$-fair if the total variation distance between the output distribution of the algorithm and the fair prior may be upper bounded by $\alpha$. 



We develop two computationally-efficient algorithms $\emix$ and $\smix$ for the above problem of fair-to-welfare interpolation. 
The first algorithm $\emix$ essentially achieves an optimal multi-criteria approximation with respect to fairness and welfare with a small number of samples drawn from the fair prior. Our second algorithm $\smix$ is even simpler with zero computational overhead. It matches the performance of $\emix$ up to constant factors while satisfying an \emph{ancillary individual fairness} property (discussed in detail in Section~\ref{sec:smix-ind}).
We summarize our main results as follows:

\begin{itemize}
    \item For any instance of FWI, our algorithm $\emix$ draws $O(1/\eps^2)$-many samples from the given fair prior and achieves  optimal welfare while being $\alpha$-fair.
    \item   For any instance of FWI, our algorithm $\smix$ requires a single sample from the fair prior to be within a constant factor from the optimal welfare while being $\alpha$-fair.
\end{itemize}

Our framework neither makes any assumptions about the underlying fair mechanism, other than that it produces a distribution over allocations, nor any assumptions about the welfare-maximizing mechanism. 
It is relevant to note that both of our algorithms are \emph{oblivious} to the welfare-maximizing mechanism. Indeed, if this mechanism is itself an approximation algorithm, our framework will retain the approximation factor for any degree of interpolation, guaranteeing sensitivity to the quality of the underlying algorithms, while being oblivious to them. Furthermore, all of our results hold for \emph{any} welfare function and \emph{any} fair prior over the solution space.

We test and complement our theoretical results on popular ex-ante notions of fairness such as randomized round-robin and sortition on real data sets.\\

\noindent
\subsection*{Additional Related work:}
Previous works have studied various notions of ex-post fairness along with welfare guarantees \cite{caragiannis2019unreasonable,aziz2020polynomial,arunachaleswaran2019fair,barman2020optimal} for the indivisible setting. Due to the hardness of computing such allocations, it is essential to relax either the fairness or the welfare goals (and sometimes maybe both). Hence, exploring ex-ante notions of fairness and welfare becomes relevant. In the economics literature, the first work to introduce the idea of finding a fractional allocation and implementing it as a lottery over pure assignments is by Hylland and Zeckhauser \cite{hylland1979efficient}. Randomness has become an important design tool for mechanisms for solving allocation problems. Subsequent works implemented this idea in various settings. Mechanisms such as the probabilistic serial rule \cite{bogomolnaia2001new} consider randomness to avoid any bias towards the agents and achieve ex-ante fairness guarantees for all cardinal valuations that are consistent with the ordinal preferences. Other such ex-ante-based mechanisms in this line of research include random priority \cite{abdulkadirouglu1998random}, vigilant eating \cite{aziz2022vigilant}, and several extensions of the probabilistic serial rule \cite{budish2013designing}.

The works of Freeman et al. \cite{freeman2020best} and Aziz \cite{aziz2020simultaneously} study the possibility of achieving ex-ante and ex-post fairness guarantees simultaneously in the classical fair allocation setting along with Pareto efficiency.  There are works in other related areas including voting \cite{aziz2019probabilistic} and two-sided matching problems \cite{budish2013designing,chen2002improving}, that have also explored the concept of randomization to circumvent the non-existence of various solutions concepts.

%% file: sections/02.notation.tex
\section{Notation and Framework}
Consider a solution space $\Sol$ and a non-negative function $V:\Sol \rightarrow \mathbb{R}^{+}_0$ defined on $\Sol$. We say that the function $V$ assigns a \emph{value} or \emph{welfare}, $V(i)$ to a solution $i \in \Sol$. Further, we will use $\ou$ to denote a solution in $S$ with maximum welfare i.e. $\ou=\argmax_{i\in\Sol}V(i)$. 


In what follows we will use $\mcl$ to denote a welfare-maximizing mechanism with an approximation factor $\lambda$ with respect to $V(\ou)$. We will denote $A \in \Sol$ to be the solution provided by $\mcl$, i.e, $V(A)\geq \lambda V(\ou)$. 

Furthermore, we consider probability distribution vectors  $p =\{p_i\}_{i \in \Sol}$ over the solution space. Here $p_i \in [0,1]$ specifies the probability assigned to solution $i \in \Sol$. We write $\Delta:=\{p =\{p_i\}_{i \in \Sol}: p_i \in [0,1], \sum_i p_i=1 \}$ to denote the set of all distributions vectors over the solution space $\Sol$. For $p \in \Delta$, we linearly extend the definition of the value function and write $V(p)=\sum_{i\in \Sol} p_i \cdot V(i)$. That is, $V(p)$ is the expected value of a random variable that takes values in the solution space $\Sol$ with distribution $p$. 

We assume that a \emph{fair mechanism} can be expressed as a random variable over the space of all solutions and we call the corresponding distribution vector $\pf \in \Delta$ the \emph{fair prior}.
Given a distribution vector $p \in \Delta$ and a distribution vector $\pf$ of the fair mechanism, we model the distance of $p$ to the fair prior $\pf$ using the standard notion of total variation distance. That is, for a fair prior $\pf$ and an  $\alpha \in [0,1]$, we say that $p \in \Delta$ is $\alpha$-\emph{fair} if  $$\TV(p,\pf):= \frac{1}{2}  \sum_{i\in \Sol}  |p_i-\pf_i| \leq \alpha.$$

We now define an instance of \emph{fair-to-welfare interpolation} (FWI) with the following quadruplet $\instance$, where $V$ is a value function, $\pf$ a fair prior which can be accessed via sampling, $\mcl$ a welfare-maximizing mechanism, and $\alpha\in (0,1)$. Given an FWI instance, the goal is to develop an algorithm whose output has a distribution $p$, such that $p$ is $\alpha$-fair and maximizes the value function $V$. In other words, we aim to solve the following constrained maximization problem, whose optimal solution we will denote $\po$,
\begin{gather*}
\qquad \max_{p \in \Delta} V(p) \\
\text{such that} \ \TV(p,\pf)\leq \alpha.
\end{gather*}

%% file: sections/03-algorithm-arxiv.tex
\section{Interpolating Between Fairness and Welfare }
In this section, we develop a computationally-efficient algorithm, called $\emix$ that, for any instance of FWI, essentially achieves an optimal multi-criteria approximation with respect to fairness and welfare with a small number of samples drawn from the given ex-ante fair mechanism. Our main result of this section is stated in Theorem~\ref{thm:emix}.

\begin{theorem}\label{thm:emix}
Consider an instance  $\instance$ of FWI along with an approximation factor $\eps>0$ and let $\px$ be the distribution of the output of Algorithm~\ref{alg:cap}. Then, $\px$ is $\alpha$-fair and obtains the following welfare guarantees
\begin{align*}
V(\px) \geq \lambda \cdot (1-\eps) \cdot V(\po).
\end{align*}
using $O(1/\eps^2)$-many samples drawn from the fair prior. 
\end{theorem}



 The key idea of our algorithm $\emix$ is to redistribute mass from solutions with low welfare in the support of ex-ante fair prior $\pf$ over to the solution provided by the welfare-maximizing mechanism (with approximation factor $\lambda$). 
 The issue with doing this greedily and explicitly is that it is typically computationally infeasible: the solution space may be exponential in the size of the input and even if this is not considered an obstacle, it is computationally infeasible to exactly determine the probability of a mechanism outputting a given solution.

 Instead, we only access the fair prior distribution $\pf$ by independently drawing samples from it. For a given FWI instance, $\emix$ begins with flipping a coin that returns heads with a probability $\alpha$. If the coin flip comes up heads, then it outputs the solution provided by the welfare-maximizing mechanism $\mcl$. Otherwise, it gathers an ensemble of samples from the sampling mechanism and picks a solution that is among the top $(1-\alpha)$-fraction of the highest-valued solutions. Observe that the above-described scheme is a natural way of mimicking the process of moving mass from the $\alpha$-fraction of the smallest-valued solutions in the fair prior $\pf$ to the solution of $\mcl$, only using the sampling mechanism. However, when we are applying this sampling approach we can not guarantee that these $(1-\alpha)$-fraction of the highest-valued solutions of the ensemble are those in the top $(1-\alpha)$-fraction of the highest-valued solution in the support of $\pf$ as well. However, we show that using $O(\eps^{-2})$ many samples, $\emix$ attains a welfare value that is $(1-\eps)$-close to the optimal, i.e., the value we could have obtained with the explicit knowledge of $\pf$. We now present $\emix$ formally in Algorithm~\ref{alg:cap}.
{
\begin{algorithm}[ht]
{
\caption{Fairness-to-welfare interpolation ($\emix$)}\label{alg:cap}
 {\bf Input:} An instance $\instance$ of FWI, and an approximation factor $\eps>0$. \\
 {\bf Output:} A solution $i \in \Sol$ drawn according to a distribution $\px \in \Delta$ such that $\TV(\px,\pf) \leq \alpha$ and $V(\px)\geq \lambda \cdot (1-\eps)  \cdot V(\po)$. 
\begin{algorithmic}[1]
\STATE Flip a coin that returns heads with probability $\alpha$.
\IF{the coin flip = heads}
\STATE Run $\mcl$ and denote its output by $A \in \Sol$.
   \RETURN $A$
\ELSE \label{alg:else}
   \STATE{Sample $s:=\left\lceil 8((1-\alpha)\eps^2)^{-1}\log(2/\eps)\right\rceil$ many solutions according to the fair prior $\pf$, and  let $Z_s^{()}$ be the vector consisting of these solutions sorted in decreasing order of value.}
   \STATE{Assign one unit of mass, $w_i=1$ to each entry $i\in [s]$ of $Z_s^{()}$.}
   \STATE{Remove a total of $\alpha s$ mass starting from the last entry of $Z_s^{()}$ and moving upwards. Denote the updated weight vector by $w^{\alpha}$}.
   \RETURN  a solution from $Z_s^{()}$ chosen with probability proportional to the weight vector $w^{\alpha}$.
\ENDIF
\end{algorithmic}
}
\end{algorithm}
}

We now give our analysis of $\emix$ proving that it is $\alpha$-fair and achieves the claimed $\lambda(1-\eps)$ approximation factor compared to the optimal $V(\po)$. We will prove Theorem~\ref{thm:emix} using  the following two lemmas. The first lemma proves that $\emix$ is $\alpha$-fair.
\begin{lemma}\label{lem:fair}
For any instance,  $\instance$ of FWI, let $\px$ denote the distribution of Algorithm~\ref{alg:cap}'s output. Then, we have
\begin{align*}
  \TV(\px,\pf) = \frac{1}{2}\sum_{i} |\px_i-\pf_i|\leq \alpha.
 \end{align*}
\end{lemma}
Observe that Lemma~\ref{lem:fair} does not depend on the number of samples that $\emix$ draws from the fair prior, i.e. $\emix$ is always $\alpha$-fair. Our next lemma says that the welfare $V(\px)$ of the output distribution is at least $\lambda(1-\eps)$ times the optimal $V(\po)$.
\begin{lemma}\label{lem:aproxemix}
Consider an instance  $\instance$ of FWI. For any $\eps>0$, let $\px$ be the distribution of the output of Algorithm~\ref{alg:cap} obtained by drawing $O(1/\eps^2)$-many samples from the fair prior $\pf$. Then, we have the following welfare guarantee,
\begin{align*}
V(\px) \geq \lambda \cdot (1-\eps) \cdot V(\po).
\end{align*}
\end{lemma}
Combining Lemmas~\ref{lem:fair} and \ref{lem:aproxemix} yields Theorem~\ref{thm:emix}, and hence we move on to their proofs. 
We begin by introducing some useful notation. Let $z_s$ denote a set  of $s$ solutions i.e. $z_s\in \Sol$. We slightly abuse notation and further also consider $z_s$ as an indicator vector over these solutions, where the entries of $z_s$ are sorted in decreasing order of welfare and the $j$'th entry is referred to as $z_s^{(j)}$. Now, for any $\beta \in (0,1)$, consider the following random mechanism $\ta$. It creates a weight vector $w^\beta\in \mathbb{R}^s$ that is initially uniform. Then it removes a total of $\beta s$ mass iteratively, starting with $i=s$ and continuing towards $i=1$, creating a residual vector $w^\beta$. Given this vector $\ta$ picks a solution in $z_s^{()}$ proportional to the weight of $w^\beta$.  By the above discussion, for any $z_s$ we can write
\begin{align}\label{alg:taprop}
\p(\ta(z_s)=i)=\sum_{j=1}^s  \mathbbm{1}_{z_s^{(j)}=i}w_j^\beta /\left((1-\beta)s\right)
\end{align}
Now let $Z_s$ denote a random vector of length $s$, with i.i.d. solutions drawn according to $\pf$. Further, let $Z_s^{(j)}$ denote the $j$'th order variable of $Z_s$ entries in terms of decreasing order. Observe that for $\beta=\alpha$ and $s$ as stated in Step 6, the mechanism $\taa$ works identically to Steps 6-9 in Algorithm~\ref{alg:cap}. Using this together with Equation~(\ref{alg:taprop}) and that $A \in \Sol$ is $\lambda$-approximate welfare-maximizing solution given by $\mcl$, we can write $\px_i$ as follows
\begin{equation}\label{alg:proppx}
\px_i =  \alpha\mathbbm{1}_{i=A}+\sum_{j=1}^s\e \left[\mathbbm{1}_{Z_s^{(j)}=i}\right]\wa_j /s,
\end{equation}
and its expected value $V(\px)$ as 
\begin{align}\label{alg:valuepx}
V(\px)=\alpha V(A)+\sum_{j=1}^s\e \left[V(Z_s^{(j)})\right]\wa_j /s.
\end{align}

It is perhaps instructive to observe that for $\alpha=0$, this mechanism $\tz(z_s)$ corresponds to drawing one entry in $z_s$ uniformly at random. Combining Equations \ref{alg:taprop} and \ref{alg:proppx}, we can write $\pf_i$ as
\begin{align}\label{alg:proppf}
\pf_i= \sum_{j=1}^s\e\left[\mathbbm{1}_{Z_s^{(j)}=i}\right]/s
\end{align}

With the above notation, we are now ready to give the proof of \Cref{lem:fair} i.e. that $\emix$ is $\alpha$-fair. 

\begin{proof}[Proof of Lemma~\ref{lem:fair}]
We need to show that the ditance $\TV(\px,\pf)\leq \alpha$, which by definition is equivalent to showing $\sum_{i\in \Sol} \left|\px_i-\pf_i\right|\leq 2\alpha$. For any $i \in [s]$, using the expression of $\px_i$ and $\pf_i$ in equations~(\ref{alg:proppx}) and (\ref{alg:proppf}) respectively, along with the fact that  $\wa_j\leq w_j$ for any $j \in [s]$ and applying the triangle inequality, we get that the absolute difference between $\px_i$ and $\pf_i$ is
\begin{align*} 
 \left|\px_i-\pf_i\right| \leq \alpha\mathbbm{1}_{i=A}+\sum_{j=1}^s\e \left[\mathbbm{1}_{Z_s^{(j)}=i}\right](w_j-\wa_j) /s.
\end{align*}
Now since $\sum_{i\in \Sol} \mathbbm{1}_{Z_s^{(j)}=i}=1$, we obtain using the linearity of expectations that 
\begin{align*}
   \sum_{i\in\Sol}\sum_{j=1}^s\e \left[\mathbbm{1}_{Z_s^{(j)}=i}\right](w_j-\wa_j) /s 
   &= \alpha
\end{align*}
where the last equality follows using $\sum w_j=s$ and $\sum \wa_j=(1-\alpha)s$. Finally, combining the last two equations and using the fact that $\mathbbm{1}_{i=A}=1$ for $i=A$ and zero otherwise, we get the desired bound of $\sum_{i\in \Sol} \left|\px_i-\pf_i\right|\leq 2\alpha$. This completes our proof. 
\end{proof}

Next, we move on to prove Lemma~\ref{lem:aproxemix} that says $\emix$ achieves an $\lambda(1-\eps)$ approximation ratio compared to $V(\po)$. For this, we  introduce some more notation. Recall that $\po$ is the optimal solution to FWI, that is $\TV(\po,\pf)\leq \alpha$ and $V(\po)$ is maximized. In the following, we describe a characterization for the optimal solution $\po$. Start with $\pf$ and remove $\alpha$ mass from the solutions in the $\text{supp}(\pf)$, beginning with the smallest-valued solution and moving towards solutions with higher value. We write $p^{\alpha,r}$ to denote the residual mass left of $\pf$ after the above procedure. Observe that $\po$ is now obtained from $p^{\alpha,r}$ by adding $\alpha$ mass to entry $\ou$. In this way, $\po$ is $\alpha$-fair and it has added as much mass to the optimal solution i.e. it is indeed the optimal solution to the constrained maximization problem of FWI. Note that $\sum_{i\in\Sol}p^{\alpha,r}_i=1-\alpha$ and hence $p^{\alpha,r}/(1-\alpha)$ is a probability measure on the top $(1-\alpha)$-fraction of solutions in $\pf$ (in terms of welfare). In the following, we will use $\paa$ to denote $p^{\alpha,r}/(1-\alpha)$. Finally, $\po$ can be expressed using $\paa$ as follows
\begin{align}\label{alg:proppo}
\po_i=
\alpha\mathbbm{1}_{i=\ou}+(1-\alpha)\paa_i,
\end{align}
implying that 
\begin{align}\label{alg:valuepo0}
V(\po)=\alpha V(\ou)+(1-\alpha)V(\paa).
\end{align}

Moreover, we write $\paac$ to denote $(\pf-p^{\alpha,r})/\alpha$, i.e. a probability measure over the bottom $\alpha$-fraction of solutions in $\pf$ (in terms of welfare). Note that we may decompose $\pf$ using $\paa$ and $\paac$ in the following way
\begin{align}\label{alg:decompf}
\pf=(1-\alpha)\paa+\alpha\paac.
\end{align}

Using the above machinery, we are now ready to present a lemma that will be crucial in proving our Lemma~\ref{lem:aproxemix} that establishes the welfare guarantees for $\emix$. The following lemma says that the expected welfare of a solution picked between the $1-\alpha$-top solutions of an ensemble of $s=\left\lceil 8((1-\alpha)\eps^2)^{-1}\log(2/\eps)\right\rceil$-many samples from the fair prior is close to the expected welfare of $\paa$.

\begin{claim}\label{cor:valuepx}
Consider an instance  $\instance$ of FWI  and an approximation factor $\eps>0$. Then, we have
\begin{align*}
 \sum_{j=1}^s\e\left[V(Z_s^{(j)})\right]\wa_j /s\geq(1-\alpha)(1-\eps)V(\paa)
\end{align*}
\end{claim}

Before proving  \Cref{cor:valuepx} we will introduce some notation. Recall that $\paa=p^{\alpha,r}/(1-\alpha)$ and $\paac=(\pf-p^{\alpha,r})/\alpha$, where $p^{\alpha,r}$ was the resulting vector of removing $\alpha$ mass from the entries of $\pf$ starting from the lowest-valued entries. Using this, we noticed in \cref{alg:decompf} that we may decompose $\pf$ using $\paa$ and $\paac$ in the following way
\begin{align*}
\pf=(1-\alpha)\paa+\alpha\paac.
\end{align*}
Therefore, sampling from $\pf$ can be seen as first flipping a coin $C$ (that returns heads with probability $1-\alpha$), if the coin flips is heads sample from $\paa$, else sample from $\paac$. We will view sampling from $\pf$ in this manner. For a sample vector $Z_s$ with s i.i.d. samples from $\pf$, let $C_i$ be the coin flip used in for the $i$'th sample to determine to sample from $\paa$ or $\paac$. Further we define $C^s=\sum_{i=1}^sC_i$. That is $C^s$ is the number of times we sample from $\paa$. Furthermore we write $Z_{r_1},\ldots,Z_{r_C}$ to denote the independent samples from the distribution $\paa$ in $Z_s$ and $Z_{\bar{r}_1},\ldots,Z_{\bar{r}_{s-C}}$  the independent samples from the distribution $\paac$. Let $Z_{r}^{(1)},\ldots,Z_{r}^{(C)}$ denote the order variables of $Z_{r_1},\ldots,Z_{r_C}$ in terms of welfare in decreasing order.  With this notation, we are now ready to prove \Cref{cor:valuepx}. 
\begin{proof}[Proof of Claim~\ref{cor:valuepx}]
First recall that $Z_s^{(j)}$ was the $j$'th order variables of $Z_s$, ordered in terms of decreasing welfare. Further, note that to show \Cref{cor:valuepx} we have to argue that $\e \left[\sum_{j=1}^sV(Z_s^{(j)})\wa_j /s\right]$ is larger than $(1-\alpha)(1-\eps)V(\paa)$. Here $\wa$ was produced by assigning $1$ unit of mass to a vector of length $s$ and then removing $\alpha s$ units of mass starting from index $s$ and decreasing towards index $1$. Thus we have that $\wa_j$ is $0$ for $j> \left\lceil(1-\alpha) s\right\rceil$ i.e.,  $$\sum_{j=1}^sV(Z_s^{(j)})\wa_j /s=\sum_{j=1}^{\left\lceil(1-\alpha) s\right\rceil}V(Z_s^{(j)})\wa_j /s$$ Further since $V(Z_s^{(j)})\wa_j /s\geq0$ we also have $$\sum_{j=1}^{\left\lceil(1-\alpha) s\right\rceil}V(Z_s^{(j)})\wa_j /s\geq\sum_{j=1}^{\left\lceil(1-\alpha) (1-\eps/2)s\right\rceil}V(Z_s^{(j)})\wa_j /s $$ Therefore, if we can show that the expected value of $\sum_{j=1}^{\left\lceil(1-\alpha)(1-\eps/2) s\right\rceil}V(Z_s^{(j)})\wa_j /s$ is lower bounded by $(1-\alpha)(1-\eps)V(\paa)$ we are done.

To show this consider an outcome $c_s$ of $C_S$ such that  $c_s\geq \left\lceil(1-\alpha)(1-\eps/2)s\right\rceil$ i.e. $Z_s$ contains at least $\left\lceil(1-\alpha)(1-\eps/2)s\right\rceil$ samples from $\paac$. Since $Z_{\bar{r}_j}$ for $j=1,\ldots,s-c_s$ is an outcome from $\paac$, the $\alpha$-bottom of $\pf$ in terms of welfare, we have that $V(Z_{\bar{r}_j})\leq V(Z_{r_i})$ for any $i=1,\ldots,c_s$ and $j=1,\ldots,s-c_s$. Combining this with $c_s\geq \left\lceil(1-\alpha)(1-\eps/2 )s\right\rceil $ we may write $\sum_{j=1}^{\left\lceil(1-\alpha)(1-\eps/2 )s\right\rceil}V(Z_s^{(j)})\wa_j /s$ as $\sum_{j=1}^{\left\lceil(1-\alpha)(1-\eps /2)s\right\rceil}V(Z_{r}^{(j)})\wa_j /s$. Further notice that 
\begin{align*}
\sum_{j=1}^{\left\lceil(1-\alpha)(1-\eps/2 )s\right\rceil}V(Z_{r}^{(j)})\wa_j /s 
\geq\sum_{j=1}^{\left\lceil(1-\alpha)(1-\eps/2 )s\right\rceil}V(Z_{r_j})\wa_j /s
\end{align*}
 Thus, their expectations also follow the same relation. Using the fact that $Z_{r_j}$ is a sample from $\paa$ and that we have $\sum_{j=1}^{\left\lceil(1-\alpha)(1-\eps /2)s\right\rceil}\wa_j \geq (1-\alpha)(1-\eps/2)s$ we get that the expectation of $\sum_{j=1}^{\left\lceil(1-\alpha)(1-\eps/2 )s\right\rceil}V(Z_{r_j})\wa_j /s$ at least $(1-\alpha)(1-\eps)V(\paa)$ conditioned on $c_s$. That is we have shown for outcomes $c_s$ of $C_s$ such that $c_s\geq \left\lceil(1-\alpha)(1-\eps/2)s\right\rceil$ 
 \begin{align*}
 &\e\big[\sum_{j=1}^{\left\lceil(1-\alpha)(1-\eps/2 )s\right\rceil}V(Z_s^{(j)})\wa_j /s\\&\mid C_s\geq\left\lceil(1-\alpha)(1-\eps/2 )s\right\rceil \big]\geq (1-\alpha)(1-\eps/2)V(\paa)
 \end{align*}
Thus if we can show that $\p\left[ C_s\geq\left\lceil(1-\alpha)(1-\eps/2 )s\right\rceil\right]\geq (1-\eps/2)$ the claim follow by the law of total expectation and $(1-\eps/2)^2\geq(1-\eps)$. To see that $\p\left[ C_s\geq\left\lceil(1-\alpha)(1-\eps/2 )s\right\rceil\right]\geq (1-\eps/2)$ is the case we use Hoeffdings inequality. Notice that $E[C_s]=(1-\alpha)s$, thus Hoeffdings inequality yields that

\begin{align*}
\p\left[C_s\leq(1-\alpha)(1-\eps/2)s\right]\leq \exp(-(1-\alpha)s\eps^2/8).
\end{align*}
Now since we have for $s=\left\lceil 8((1-\alpha)\eps^2)^{-1}\log(2/\eps)\right\rceil$ that $\exp(-(1-\alpha)s\eps^2/8)\leq \eps/2$ we conclude that with probability at least $1-\eps/2$, $C_s$ is strictly larger than $ (1-\alpha)(1-\eps/2)s$ and since $C_s$ is an integer $C_s$ is then also larger than $\left\lceil(1-\alpha)(1-\eps/2)s\right\rceil$, which concludes our proof.
\end{proof}

 We now present the proof of Lemma~\ref{lem:aproxemix} that says $\px$ achieves welfare guarantees with an approximation factor of $(1-\eps)\lambda$ as compared to $\po$. 

\begin{proof}[Proof of Lemma~\ref{lem:aproxemix}]
We need to show that $V(\px)$ is lower bounded by $(1-\eps)\lambda V(\po)$. To begin with, using equation~(\ref{alg:valuepx}) we can express $V(\px)$ as 
\begin{align*}
V(\px)=\alpha V(A)+\sum_{j=1}^s \e \left[V(Z_s^{(j)})\right]\wa_j /s.
\end{align*}
Using Lemma~\ref{cor:valuepx} and since $A$ is the solution provided by $\mcl$ i.e., $V(A)\geq \lambda V(\ou)$,  it follows that 
\begin{align}\label{alg:valuepx1}
V(\px)&\geq\alpha \lambda V(\ou)+(1-\alpha)(1-\eps)V(\paa)\nonumber \\
&\geq (1-\eps)(\alpha \lambda V(\ou)+(1-\alpha)V(\paa)).
\end{align}

Recalling the expression for $V(\po)$ from equation~(\ref{alg:valuepo0}) for comparison.
\begin{align}\label{alg:valuepo1}
V(\po)=\alpha  \cdot V\left(\ou\right)+(1-\alpha)  \cdot V\left(\paa\right).
\end{align}

Now, note that, for $a=\alpha  \cdot V\left(\ou\right) $, $b=\alpha \lambda \cdot V \left(\ou\right) $ and $x=(1-\alpha)V(\paa)$, we have $a\geq b$ and $x\geq0 $. Algebraic manipulations then lead to the fact that we must have $\frac{b+x}{a+x}\geq \frac{b}{a}$. Now combining this fact with equations~(\ref{alg:valuepx1}) and (\ref{alg:valuepo1}), we obtain that $\frac{V(\px)}{V\left(\po\right)}$ is lower bounded by 
\begin{align*}
\frac{(1-\eps)(\alpha \lambda \cdot V \left(\ou\right) +(1-\alpha)V(\paa))}{\alpha  \cdot V\left(\ou\right)+(1-\alpha)V\left(\paa\right)}
\geq (1-\eps)\lambda,
\end{align*}
which concludes the proof. 
\end{proof}
\section{Simple Mixing to achieve Fairness and Welfare guarantees}
In this section, we present a faster and simpler algorithm $\smix$ for the problem of fair-to-welfare interpolation (FWI). We show that  $\smix$ is $\alpha$-fair and achieves welfare guarantees that are optimal up to a constant factor. An important feature of $\smix$ is that it has absolutely no computational overhead.
Additionally, we will present the ancillary individual fairness properties beyond the underlying welfare maximizing allocation and the ex-ante fair mechanism at the end of this section.

 For a given FWI instance, $\smix$ proceeds by flipping a coin that returns head with probability $\alpha$. If the coin comes up heads then the algorithm outputs the solution provided by the welfare-maximizing mechanism $\mcl$ else it outputs the solution drawn from the fair sample mechanism. Note that, this can be interpreted as an interpolation between the $\lambda$-approximation solution with $\alpha$ mass and the fair prior distribution scaled with a factor of $1-\alpha$, see  Algorithm~\ref{alg:2approx} below.
. We will show that $\smix$ is $\alpha$-fair and attains an approximation factor of $\min\{\lambda,\alpha \lambda+(1-\alpha)^2\}$. 
 Unlike the $emix$ algorithm, $\smix$ has no computational overhead and in particular, it does not have to evaluate the value of any solution.

\input{code/epsilon_mix.tex}
We now present the main theorem of this section stating that $\smix$ is $\alpha$-fair and obtain an approximation factor of $\min\{\lambda,\alpha \lambda+(1-\alpha)^2\}$ compared to the welfare of $\po$

\begin{theorem} \label{thm:smix}
Consider an instance  $\instance$ of FWI and let $\pa$ be the distribution of the output of Algorithm~\ref{alg:2approx}. Then, $\pa$ is $\alpha$-fair and achieves the following welfare guarantee
\begin{align*}
V(\pa) \geq \min\{\lambda,\alpha \lambda+(1-\alpha)^2\} \cdot V(\po).
\end{align*}
\end{theorem}

\begin{remark} \label{rem:smix-welfare}
    Let us consider the ratio of $\min\{\lambda,\alpha \lambda+(1-\alpha)^2\}/\lambda$ and note that it is minimized at $\lambda=1$, and $\alpha=1/2$, i.e. the ratio is lower bounded by $\min\{1,3/4\}=3/4$. And hence, using Theorem~\ref{thm:smix}, we obtain that $\smix$ always achieves the following performance guarantees
    $$V(\pa) \geq \frac{3}{4} \lambda \cdot V(\po)$$
    of being optimal up to constant factors.
\end{remark}

Before proving Theorem~\ref{thm:smix}, we make a few observations about Algorithm~\ref{alg:2approx}. Recall that it returns the solution $A$ (given by the mechanism $\mcl$) with probability $\alpha$ and for the remaining probability, it returns a solution provided by the sampling mechanism for $\pf$. Therefore, for any $i \in \Sol$ we can express the distribution $\pa$ of $\smix$ as follows
\begin{align}\label{alg:proppa}
\pa_i=
\alpha\mathbbm{1}_{i=A}  +(1-\alpha )\pf_i
\end{align}
and, this implies that the expected welfare $V(\pa)$ achieves by $\smix$ may be expressed as
\begin{align}\label{alg:valuepa1}
V(\pa)=\alpha V(A)+(1-\alpha)V(\pf).
\end{align}
Combining this with the fact that $V(A)\geq \lambda V(\ou)$ we  get the following lower bound on $V(\pa)$ 
\begin{align}\label{alg:valuepa0}
V(\pa)\geq \alpha \lambda V(\ou)+(1-\alpha)V(\pf).
\end{align}

Having the above observations, we are now ready to prove Theorem~\ref{thm:smix}.

\begin{proof}[Proof of Theorem \ref{thm:smix}]
We begin by showing that $\pa$ is $\alpha$-fair, i.e. $ \TV(\pa,\pf)\leq \alpha$. By the definition of total variation distance, it suffices to show that $\sum_{i \in \Sol} \left| \pa_i-\pf_i \right|\leq 2\alpha$. We first notice that equation~(\ref{alg:proppa}) implies that $\pa_A\geq \pf_A$ and $\pa_i\leq \pf_i$ for $i\not=A$ i.e., we have $|\pa_A-\pf_A|=\pa_A-\pf_A$ and for $i\not=A$ $|\pa_A-\pf_A|=\pf_A-\pa_A$. We can therefore write
\begin{align*}
&\ \ \ \sum_{i \in \Sol} \left| \pa_i-\pf_i \right|
=\alpha  +\left(1-\alpha \right)
\cdot\pf_{A}-\pf_{A}\\ +& \sum_{i \in \Sol \setminus \{A\}}\pf_i-\left(1-\alpha \right)\pf_i\leq \alpha +  \sum_{i \in \Sol \setminus \{A\}} \alpha\pf_i\leq2\alpha
\end{align*}
proving that $\pa$ is $\alpha$-fair. 

For proving the welfare guarantees of $\smix$, we will show that the ratio $V(\pa)/V(\po)$ is at least as high as $\min\{\lambda,\alpha \lambda+(1-\alpha)^2\}$. We will split this analysis for the ratio into two cases. In the first case, we consider the scenario when $V(\pf)\geq (1-\alpha )V(\ou)$. Using the lower bound of $V(\pa)$ in equation~(\ref{alg:valuepa0}) and the fact that $V(\ou)\geq V(\po)$ we obtain that $V\left(\pa\right)/V\left(\po\right)$ is lower bounded by
\begin{align*}
&\frac{\alpha \lambda V(\ou)+(1-\alpha )^2V(\ou)}{V\left(\po\right)}\geq \alpha \lambda+(1-\alpha)^2.
\end{align*} 
as desired. Therefore, we move to the second case where $V(\pf)\leq (1-\alpha) V(\ou)$. Recall that by using equation~ (\ref{alg:valuepo0}), we can write $V(\po)=\alpha V(\ou)+(1-\alpha)\paa$. Further recall that equation~(\ref{alg:decompf}) says that we can express $\pf=(1-\alpha)\paa+\alpha\paac$ implying that we have $\paa=(\pf-\alpha\paac)/(1-\alpha)$. Now combining the above, we obtain
\begin{align*}
    V(\po)&=\alpha V(\ou)+V(\pf)-\alpha V(\paac)\\
    &\leq \alpha V(\ou)+V(\pf)
\end{align*}
 Since $V(\pa) \leq \alpha \lambda V(\ou)+(1-\alpha)V(\pf)$, we get that the ratio of $V(\pa)/V(\po)$ is lower bounded by
 \begin{align}\label{alg:analytic}
 \frac{\alpha \lambda \cdot V(\ou)+(1-\alpha )\cdot V(\pf)}{\alpha \cdot V(\ou)+V(\pf)}.
 \end{align}
We will now analyze the above expression as a function of $V(\pf)$ and show that it is lower bounded by $\min\{\lambda,\alpha \lambda+(1-\alpha)^2\}$. Note that this will be sufficient to complete our proof. Hence, we move on to show that the expression in equation~(\ref{alg:analytic}) has the desired lower bound for the case when $V(\pf)\leq (1-\alpha) V(\ou)$. 

To this end, consider the function $h(x)=\frac{\alpha \lambda\cdot V(\ou)+(1-\alpha ) x}{\alpha  \cdot V(\ou)+x}$, and its derivative $h'(x)$, as follows
\begin{align*}
\frac{\alpha (1-\alpha) \cdot V(\ou)-\alpha \lambda \cdot  V(\ou)}{\left(\alpha \cdot V(\ou)+x\right)^2}
\end{align*}
We notice that the function $h(\cdot)$ is increasing for $\lambda\leq(1-\alpha)$ and decreasing otherwise. If $h(\cdot)$ is increasing, then $h(x)\geq h(0)= \lambda$ holds true, otherwise if $h(\cdot)$ is decreasing, we have 
$h(x)\geq h((1-\alpha )V(\ou))=\alpha \lambda+(1-\alpha)^2 $ for $x\leq (1-\alpha)  V(\ou)$. This completes our proof and the stated claim holds true.
\end{proof}

\subsection{Welfare guarantees of \texorpdfstring{\smix}{Lg} are tight}\label{sec:tight-smix}

In this section, we prove that for $\alpha \in [0,1]$ and $\lambda \leq 1$ such that either (i) $\alpha\leq 1-\lambda$ or (ii) $\lambda=1$, there exists an instance $\instance$ of FWI for which the approximation factor on the welfare guarantees of $\smix$, as shown in Theorem~\ref{thm:smix} is tight.

Let us begin with the first case where $\alpha\leq 1-\lambda$. First, note that, for this case we have $\min\{\lambda,\lambda\alpha+(1-\alpha)^2\}=\lambda$ and hence we need to show that $V(\pa)=\lambda V(\po)$. We now consider the following FWI instance where the fair prior $\pf$ is such that every solution in its support has a welfare value of $0$ and $V(A)=\lambda \cdot V(\ou)$ for the output solution $A$ of the mechanism $\mcl$. Note that, for such an instance, we have $V(\pf)=0$ and $ V(\paa)=0$. Therefore, using equation~(\ref{alg:valuepo0}), we obtain that $V(\po)=\alpha V(\ou)$. Moreover, using equation~(\ref{alg:valuepa1}) and $V(A)=\lambda \cdot V(\ou)$, we also get that $V(\pa)=\alpha \lambda  V(\ou)=\lambda  V(\po)$, as desired. 

We now consider the second case where $\lambda=1$ and $0\leq\alpha\leq1$. 
First, note that for this case, we have $\min\{\lambda,\lambda\alpha+(1-\alpha)^2\}=\alpha+(1-\alpha)^2$ i.e. to show desired tightness, it suffices to construct an instance of FWI such that $V(\pa)=(\alpha+(1-\alpha)^2)V(\po)$. To this end, consider an instance with two solutions $S_1$ and $S_2$ with $V(S_1)=1$ and $V(S_2)=0$. Further, let us assume that the fair prior $\pf$ puts $\beta$ mass on $S_1$ and $1-\beta$ mass on $S_2$ for some $\beta\in[0,1]$. And hence, we have $V(\pf)=\beta$. Since, $\lambda=1$, we know that the mechanism $\mcl$ must outputs the solution $S_1$. This implies that we can write $V(\pa)=\alpha+(1-\alpha)\beta$ and $V(\po)=\min\{\alpha+\beta,1\}$. Now for $\beta=1-\alpha$ we get that $V(\po)=1$ and $V(\pa)=\alpha+(1-\alpha)^2=(\alpha+(1-\alpha)^2)V(\po)$, as desired. Since the above holds for any $\beta \in [0,1]$, it must also hold for any $\alpha \in [0,1]$ as well, thereby proving the stated claim.

\subsection{Ancillary fairness properties of \texorpdfstring{$\smix$}{Lg}}\label{sec:smix-ind}
For any instance $\instance$ of FWI, the distribution $\pa$ of the output of our algorithm $\smix$ has the following additional property: for any solution $i \in \Sol$ that is in the support of the fair prior $\pf$, i.e., $\pf_i>0$, $i$ is also in the support of $\pa$. This is also the case for $\emix$. However $\smix$ actually, also has the following stronger property that
\begin{align} \label{eq:ind-fair}
\pa_i \geq (1-\alpha) \cdot \pf_i
\end{align}for every solution $i \in \Sol$, with equality for $i\not=A$. 
This follows from the fact that $\smix$ outputs a sample drawn from the fair prior with probability $1-\alpha$. Therefore, in addition to being $\alpha$-fair, $\smix$ preserves the support of the given fair prior as well, hence providing \emph{individual fairness} guarantees. 

Now consider a setting where our FWI framework models the classic problem of allocating a set $[m]$ of $m$ indivisible items to a set $[n]$ of $n$ agents. We say an allocation is an $n$-partition of items into $n$ bundles, one for each agent. Hence, the set of solution space $\Sol$ consists of all possible allocations. An agent $a \in [n]$ has a valuation function $u_a: \Sol \rightarrow \mathbb{R}^+_0$ that defines her utility $u_a(i)$ from an allocation $i \in \Sol$. Note that a distribution vector $p \in \Delta$ is now a random allocation where an allocation $i$ is selected with probability $p_i$. With slight abuse of notation, for any $p \in \Delta$, we say her utility $u_a(p) = \sum_{i \in \Sol} p_i u_a(i)$ is the expected utility she derives from the associated random allocation to the distribution $p$. Finally, the value function $V$ can model social welfare, $V(p)=\sum_{a \in [n]}u_a(p)$ or Nash social welfare $V(i)= (\prod_{a \in [n]} u_a(p))^{1/n}$ of a random allocation associated with $p \in \Delta$. 

Using equation~(\ref{eq:ind-fair}), we conclude that for any agent $a \in [n]$, her utility in the distribution $\pa$ of $\smix$ is 
\begin{align*}
    u_a(\pa) &= \sum_{i \in \Sol} \pa_i \cdot u_a(i) \\
    & \geq (1- \alpha) \sum_{i \in \Sol} \pf_i \cdot u_a(i) = (1- \alpha) \cdot u_a(\pf)
\end{align*}
at least $(1-\alpha)$ times her utility in the fair prior, hence providing \emph{individual fairness} guarantees in addition to being $\alpha$ to fairness.

%% file: code/epsilon_mix.tex
\begin{algorithm}[ht]
\caption{Fairness-to-welfare interpolation ($\smix$)}\label{alg:2approx}
 {\bf Input:}An instance $\instance$ of FWI. \\
 {\bf Output:} A solution $i \in \Sol$ drawn according to a distribution $\pa \in \Delta$ such that $\TV(\pa,\pf) \leq \alpha$ and $V\left(\pa\right)\geq \min\{\lambda,\alpha \lambda+(1-\alpha)^2\}V\left(\po\right)$.
\begin{algorithmic}[1]

\STATE Flip a coin that returns heads with probability $\alpha $.
\IF{the coin flip = heads}
    \STATE Run  $\mcl$, and denote its output by $A \in \Sol$.
   \RETURN $A$
\ENDIF
\RETURN{a solution sampled according to the fair prior $\pf$.}
\end{algorithmic}
\end{algorithm}

%% file: sections/04-experiments-arxiv.tex
\section{Experiments}

We complement our theoretical analysis with an extensive practical evaluation.  
$\emix$ should theoretically produce better welfare than $\smix$. Nevertheless, the simplicity of $\smix$, as well as its ancillary individual fairness property begs the question of whether the theoretical performance gap between the two algorithms also exists in practice, or whether $\smix$ is the preferable algorithm.
We conduct an experimental study on two main scenarios: \emph{assignments} and \emph{sortition}. Both have natural ex-ante fair priors and represent practical use cases of our framework. We considered FWI instances both reflecting the ex-ante fair mechanism randomized round-robin (RRR) and sortition. 

For RRR, we considered the problem of matching papers to reviewers using the AAMAS  data set \cite{PrefLib_AAMAS}. Welfare is measured in terms of average reviewer preference for their assigned papers. 
Specifically, we model these problems as weighted $B$-matching in a bipartite graph with one side of the bipartition corresponding to reviewers and the other side corresponding to papers. 
RRR initially samples a random permutation. The reviewers then choose their papers greedy from the top of the available preference list according to the order of the permutation. The selection is repeated in a round-robin fashion until every paper has sufficiently many reviewers.

For sortition, we used the Adults data set \cite{UCI_adult}, and following the line of work by \cite{ebadiansortition}, we modeled welfare as voter representation in the generated committee of size $k$. Each individual is viewed as a point in Euclidean space. Representation is modeled as the (squared Euclidean) distance between an individual and its closest representative. The value of a committee of size $k$ is the likelihood of a mixture model of $k$ Gaussians each with an identity covariance matrix centered around the points associated with members of the selected committee. The sortition mechanism starts by selecting an initial committee such that the likelihood is maximized, followed by  replacing some committee members with randomly selected individuals.

Overall, it is  possible to generate a (synthetic) instance where there is a detectable gap between $\emix$ and $\smix$. For real-world instances, this gap typically vanishes, see Figure \ref{fig:means}.
Thus, we conclude that the $\smix$ algorithm, due to its simplicity, speed, and ancillary fairness property is, at least empirically, the algorithm to use.
\input{tikz/score_means.tex}

\input{tikz/score_var.tex}

\subsection{Data sets}
For the bidding scenarios, we use a synthetic data set and a real-world bidding data set \cite{PrefLib_AAMAS}. For the sortition scenario, we use demographic data from UCI Adults \cite{UCI_adult}.

\textit{Synthetic:} 
We used a proof of concept synthetic data set that mimics a setup where $N_R$-many goods have to be distributed between $N_L$-many agents. The data is an outcome of $N_R\cdot N_L$-many uniform random variables on the interval $[0,1]$. These outcomes were now taken to be the weights of a fully connected weighted bipartite graph with $N_L$ nodes on the left and $N_R$ nodes on the right. For values of $N_R$ and $N_L$ can be seen in \Cref{tab:Datasets}.


\textit{AAMAS:} These two bidding data sets include partially the 2015 and 2016 reviewer's preference toward paper. The data consisted of each reviewer's preference list with labels "yes", "maybe", "no response" or "no". From this data, we created a complete weighted bipartite graph with the reviewers $N_L$ as left nodes and the $N_R$ papers as the right nodes (see \Cref{tab:Datasets} for the value of $N_R$ and $N_L$). The weights of the edges were decided by the label, respectively weighted $\{1,\sfrac{1}{2},0,0\}$.


\textit{UCI-Adult:} This is a typical data set used in the sortition scenarios. The data set contains a range of demographic features related to individuals. We choose the features \fet{age}, \fet{education year}, \fet{marital status}, \fet{relationship}, \fet{hours per week}, \fet{race} and \fet{sex}. After reducing duplicates, we retain $N_P=17749$ distinct individuals. We now defined a metric between points. We $1$-hotted the categorical features and interpreted the numerical features as coordinates in Euclidean space similar to \cite{ebadiansortition}.

\input{table/datasets.tex}

\subsection{Evaluation Details}
We ran $\smix$ and $\emix$ for $\alpha\in \{1/20,2/20,\ldots,19/20\}$ (and $\eps=0.1, 0.05, 0.01$). For each value of $\alpha$ (and $\eps$) $\smix$ and $\emix$ we repeated the experiment $N_{rounds}$ times and reported the empirical mean. Additionally, we reported the standard deviation of the empirical mean estimator over $N_{batch}$ runs(See \Cref{fig:std} for standard deviations). Further, the number of samples $N_{\eps}$ from the fair prior used by $\emix$ were only depending on $\eps$ for the Goods Bidding data. For the other cases, we choose a fixed sample number $N_\eps$. \Cref{tab:expsetup} displays the numbers $N_{rounds}$, $N_{batch}$ and $N_{\eps}$ in detail.
\input{table/parameter.tex}

\subsubsection{FWI Instance: Resource allocation}
In this data set, the purpose is to assign a set of limited items to a set of agents. That is, each item must have one in-going edge, which defines our solution space $\Sol$. The value function $V$ was calculated as the sum of a set of 5 such edges. $\mmb$ \cite{10.1145/6462.6502} as welfare-maximizing mechanism $\mcl$, has an approximation factor of $1$. The fair mechanism $\pf$ was made by first choosing without replacement $N_R$ agents out of the $N_L$. Then, each chosen agent paired goods in order of their preference list.


\textbf{Result:} We refer to the first plot titled, Synthetic, in Figure~\ref{fig:means}. We see that the empirical mean of $\emix$ is higher than that of $\smix$ as expected from the analysis. We further notice that the gap between $\smix$ and $\emix$ increases for $\eps$ decreasing as expected, that is as $\emix$ converges towards the optimal solution.

\subsubsection{FWI Instance: Reviewers Matching}

To define the solution space $\Sol$ for the paper-matching data sets, we set constrain for the solution as: 1) each paper should get three reviewers.  2) the difference between the number of papers two reviewers got should be no more than one. The latter constraint holds if and only if, when reviewers receive no more than $\lceil 3\cdot N_R/N_L\rceil$ papers. Therefore, each reviewer will get no more than nine papers in AAMAS 2015, and ten in AAMAS 2016.  The value function $V$ was taken to be the sum of all selected edges in such a solution.


The welfare-maximizing mechanism $\mcl$ use $\grd$. This $\grd$ algorithm, proceeded as follows: Sorted all the edges by their weights. Picked the edge with the highest weight, if it still obeys the two constraints above. 


The fair prior was taken to be repeating the following procedure until each paper is assigned three reviewers: draw without replacement all the reviewers. Following the order in which the reviewers were drawn, we let the reviewers pick the paper they liked the most among the set of papers that didn't have three reviewers already.


\textbf{Result:} We refer to the second plot titled, AAMAS15, in Figure~\ref{fig:means}. Here, we see that the plot of the empirical means of $\emix$ and $\smix$ indicates that $\emix$ slightly out preforms $\smix$.  This pattern of slight outperformance of $\emix$ for the AMMAS16 graph is more muted if even present.

\subsubsection{FWI Instance: Sortition}
In this task, we aim to find a group of $N_K$ people representing the general public. Thus we define the solution space $\Sol$ to be the family of subsets consisting of $N_K$ individuals out of the entire set. The value function is described as follows.
We measure representation via the $N_K$-means objective, where, given a set $K$ of $N_K$ representatives, we consider the cost function $\log (\mathcal{L} (K|P)) = \sum_{p\in P}\min_{k\in K}\|p-k\|^2$, derived as the negative log likelihood of a mixture of $N_K$-univariate Gaussians with the the value function $$\displaystyle \mathcal{L}(K|P)) := \prod_{p\in P} \exp(-\min_{k\in K}\|p-k\|^2)\cdot C,$$
where $C$ is an absolute constant.
\input{tikz/sortition_center.tex}

We use $\rrq$ \cite{ebadiansortition}, a sortition mechanism, as fair prior $\pf$ and $\km$ \cite{Kmeanpp} as welfare-maximizing mechanism $\mcl$. We initialize the panel for $\rrq$ by $\km$ and set $q=\lfloor{\sfrac{3}{4}\cdot N_K}\rfloor$. Subsequently, we select $q$ centers and replace a chosen center $c$ with a point chosen uniformly at random from the $q$ closest neighbors of $c$.
We considered all ranges of $N_K$ from $1$ to $2000$. As $N_K$ increases, the evaluation score gradually approaches $1$, as Figure \ref{fig:cluster}. For the purpose of testing the fairness mechanism, in the $\alpha$ interval test in Figure \ref{fig:means}, we fixed $N_K=600$.

\textbf{Result:} 
In terms of representation, purely optimizing the likelihood always yielded better values than the sortition mechanism. Nevertheless, the sortition mechanism added input points into the set of candidate centers that the optimization algorithm consistently avoided, thus balancing out representation with the chance of certain individuals to be part of the selected set. As was the case with the other algorithms, our mechanisms seamlessly interpolated between the two solutions.

%% file: tikz/score_means.tex
\begin{figure*}[ht]
\hspace*{7mm}
\begin{tikzpicture}[scale=0.45]
\begin{axis}[
    title={Synthetic},
    title style={yshift=-3 pt,},
    title style = {font=\large},
    xlabel = {$\alpha$},
    ylabel = {Empirical Mean Score},
    table/col sep=comma,
    legend cell align = left,
    legend pos = north west,
    legend style={nodes={scale=1.3, transform shape},at={(1.45,1.3)}},
    grid = major,
    grid style=dashed,
    legend columns=4,
    xlabel style={font=\large},
    ylabel style={font=\large},
    legend cell align={left},
] 
\addplot[     line width=0.2mm, 
    color=blue,
    mark=triangle,
    mark options=solid,
    mark size=2pt
    ] table[x=alpha,y=means]{tikz/plot/change_beta/synthetic_simple_1_0.1_10batches.csv};
\addplot[     line width=0.2mm, 
    color=red,
    mark=square,
    mark options=solid,
    mark size=2pt
    ] table[x=alpha,y=means]{tikz/plot/change_beta/synthetic_eps_1_0.1_5batches.csv};
\addplot[     line width=0.2mm, 
    color=RePu04,
    mark=|,
    mark options=solid,
    mark size=3pt
    ] table[x=alpha,y=means]{tikz/plot/change_beta/synthetic_eps_1_0.05_5batches.csv}; 
\addplot[     line width=0.2mm, 
    color=RePu06,
    mark=pentagon,
    mark options=solid,
    mark size=2pt
    ] table[x=alpha,y=means]{tikz/plot/change_beta/synthetic_eps_1_0.01_5batches.csv};
\legend{$\smix$, $0.1-Mix$,$0.05-Mix$,$0.01-Mix$}
\end{axis}
\end{tikzpicture}
\hspace*{-69mm}
\begin{tikzpicture}[scale=0.45]
\begin{axis}[
    title={AAMAS 2015},
    title style={yshift=-3 pt,},
    title style = {font=\normalsize},
    xlabel = {$\alpha$},
    table/col sep=comma,
    legend cell align = left,
    legend pos = north west,
    grid = major,
    grid style=dashed,
    legend style={nodes={scale=1.3, transform shape},at={(0.45,1.55)}},
    legend columns=2,
    xlabel style={font=\normalsize},
    ylabel style={font=\normalsize},
    legend cell align={left},
] 
\addplot[     line width=0.2mm, 
    color=blue,
    mark=triangle,
    mark options=solid,
    mark size=1pt
    ] table[x=alpha,y=means]{tikz/plot/change_beta/AAMAS15_simple_50_0.1_5batches.csv};
\addplot[     line width=0.2mm, 
    color=red,
    mark=square,
    mark options=solid,
    mark size=1pt
    ] table[x=alpha,y=means]{tikz/plot/change_beta/AAMAS15_eps_50_0.1_5batches.csv};
\end{axis}
\end{tikzpicture}
\begin{tikzpicture}[scale=0.45]
\begin{axis}[
    title={AAMAS 2016},
    title style={yshift=-3 pt,},
    title style = {font=\large},
    xlabel = {$\alpha$},
    table/col sep=comma,
    legend cell align = left,
    legend pos = north west,
    legend style={nodes={scale=1.3, transform shape},at={(-0.45,1.3)}},
    legend columns=4,
    grid = major,
    grid style=dashed,
    xlabel style={font=\normalsize},
    ylabel style={font=\normalsize},
    legend cell align={left},
] 
\addplot[     line width=0.2mm, 
    color=blue,
    mark=triangle,
    mark options=solid,
    mark size=1pt
    ] table[x=alpha,y=means]{tikz/plot/change_beta/AAMAS16_simple_50_0.1_5batches.csv};
\addplot[     line width=0.2mm, 
    color=red,
    mark=square,
    mark options=solid,
    mark size=1pt
    ] table[x=alpha,y=means]{tikz/plot/change_beta/AAMAS16_eps_50_0.1_5batches.csv};
\end{axis}
\end{tikzpicture}
\begin{tikzpicture}[scale=0.45]
\begin{axis}[
    title={Adult},
    title style={yshift=-3 pt,},
    title style = {font=\normalsize},
    xlabel = {$\alpha$},
    table/col sep=comma,
    legend cell align = left,
    legend pos = north west,
    grid = major,
    grid style=dashed,
    legend style={nodes={scale=1.3, transform shape},at={(0.45,1.55)}},
    legend columns=2,
    xlabel style={font=\normalsize},
    ylabel style={font=\normalsize},
    legend cell align={left},
] 
\addplot[     line width=0.2mm, 
    color=blue,
    mark=triangle,
    mark options=solid,
    mark size=1pt
    ] table[x=alpha,y=means]{tikz/plot/adult_simple_50_0.1_5batches.csv};
\addplot[     line width=0.2mm, 
    color=red,
    mark=square,
    mark options=solid,
    mark size=1pt
    ] table[x=alpha,y=means]{tikz/plot/adult_eps_50_0.1_5batches.csv};
\end{axis}
\end{tikzpicture}
\caption{ The mean score for different scenarios of FWI with $\alpha\in(0,1]$ }
\label{fig:means}
\end{figure*}
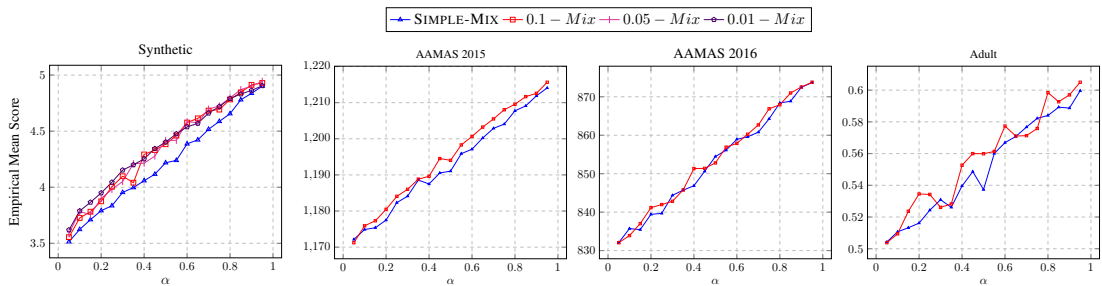

%% file: tikz/score_var.tex
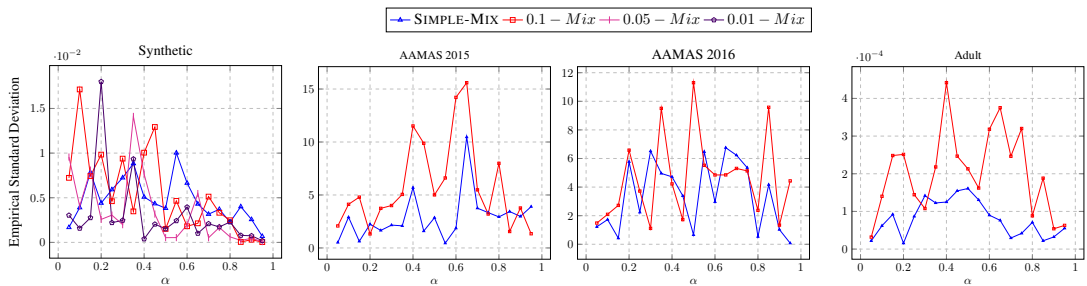
\begin{figure*}[ht!]
\hspace*{7mm}
\begin{tikzpicture}[scale=0.45]
\begin{axis}[
    title={Synthetic},
    title style={yshift=-3 pt,},
    title style = {font=\large},
    xlabel = {$\alpha$},
    ylabel = {Empirical Standard Deviation},
    table/col sep=comma,
    legend cell align = left,
    legend pos = north west,
    legend style={nodes={scale=1.3, transform shape},at={(1.45,1.3)}},
    grid = major,
    grid style=dashed,
    legend columns=4,
    xlabel style={font=\large},
    ylabel style={font=\large},
    legend cell align={left},
] 
\addplot[     line width=0.2mm, 
    color=blue,
    mark=triangle,
    mark options=solid,
    mark size=2pt
    ] table[x=alpha,y=variance]{tikz/plot/change_beta/synthetic_simple_1_0.1_10batches.csv};
\addplot[     line width=0.2mm, 
    color=red,
    mark=square,
    mark options=solid,
    mark size=2pt
    ] table[x=alpha,y=variance]{tikz/plot/change_beta/synthetic_eps_1_0.1_5batches.csv};
\addplot[     line width=0.2mm, 
    color=RePu04,
    mark=|,
    mark options=solid,
    mark size=3pt
    ] table[x=alpha,y=variance]{tikz/plot/change_beta/synthetic_eps_1_0.05_5batches.csv}; 
\addplot[     line width=0.2mm, 
    color=RePu06,
    mark=pentagon,
    mark options=solid,
    mark size=2pt
    ] table[x=alpha,y=variance]{tikz/plot/change_beta/synthetic_eps_1_0.01_5batches.csv};
\legend{$\smix$, $0.1-Mix$,$0.05-Mix$,$0.01-Mix$}
\end{axis}
\end{tikzpicture}
\hspace*{-70mm}
\begin{tikzpicture}[scale=0.45]
\begin{axis}[
    title={AAMAS 2015},
    title style={yshift=-3 pt,},
    title style = {font=\normalsize},
    xlabel = {$\alpha$},
    table/col sep=comma,
    legend cell align = left,
    legend pos = north west,
    grid = major,
    grid style=dashed,
    legend style={nodes={scale=1.3, transform shape},at={(0.45,1.55)}},
    legend columns=2,
    xlabel style={font=\normalsize},
    ylabel style={font=\normalsize},
    legend cell align={left},
] 
\addplot[     line width=0.2mm, 
    color=blue,
    mark=triangle,
    mark options=solid,
    mark size=1pt
    ] table[x=alpha,y=variance]{tikz/plot/change_beta/AAMAS15_simple_50_0.1_5batches.csv};
\addplot[     line width=0.2mm, 
    color=red,
    mark=square,
    mark options=solid,
    mark size=1pt
    ] table[x=alpha,y=variance]{tikz/plot/change_beta/AAMAS15_eps_50_0.1_5batches.csv};
\end{axis}
\end{tikzpicture}
\begin{tikzpicture}[scale=0.45]
\begin{axis}[
    title={AAMAS 2016},
    title style={yshift=-3 pt,},
    title style = {font=\large},
    xlabel = {$\alpha$},
    table/col sep=comma,
    legend cell align = left,
    legend pos = north west,
    legend style={nodes={scale=1.3, transform shape},at={(-0.45,1.3)}},
    legend columns=4,
    grid = major,
    grid style=dashed,
    xlabel style={font=\normalsize},
    ylabel style={font=\normalsize},
    legend cell align={left},
] 
\addplot[     line width=0.2mm, 
    color=blue,
    mark=triangle,
    mark options=solid,
    mark size=1pt
    ] table[x=alpha,y=variance]{tikz/plot/change_beta/AAMAS16_simple_50_0.1_5batches.csv};
\addplot[     line width=0.2mm, 
    color=red,
    mark=square,
    mark options=solid,
    mark size=1pt
    ] table[x=alpha,y=variance]{tikz/plot/change_beta/AAMAS16_eps_50_0.1_5batches.csv};
\end{axis}
\end{tikzpicture}
\hspace*{2mm}
\begin{tikzpicture}[scale=0.45]
\begin{axis}[
    title={Adult},
    title style={yshift=-3 pt,},
    title style = {font=\normalsize},
    xlabel = {$\alpha$},
    table/col sep=comma,
    legend cell align = left,
    legend pos = north west,
    grid = major,
    grid style=dashed,
    legend style={nodes={scale=1.3, transform shape},at={(0.45,1.55)}},
    legend columns=2,
    xlabel style={font=\normalsize},
    ylabel style={font=\normalsize},
    legend cell align={left},
] 
\addplot[     line width=0.2mm, 
    color=blue,
    mark=triangle,
    mark options=solid,
    mark size=1pt
    ] table[x=alpha,y=variance]{tikz/plot/adult_simple_50_0.1_5batches.csv};
\addplot[     line width=0.2mm, 
    color=red,
    mark=square,
    mark options=solid,
    mark size=1pt
    ] table[x=alpha,y=variance]{tikz/plot/adult_eps_50_0.1_5batches.csv};
\end{axis}
\end{tikzpicture}
\caption{ Empirical standard deviations for the corresponding plots in Figure~\ref{fig:means}}
\label{fig:std}
\end{figure*}

%% file: table/datasets.tex
\setlength{\extrarowheight}{1.7pt}
\begin{table}[ht]
    \centering
    \begin{tabular}{lccc}
        \toprule
          \textbf{Senario} & \multicolumn{3}{c}{\textbf{Dataset}} \\
        \midrule
        \multirow{4}{*}{Assignment}&&$N_{L}$ & $N_{R}$                  \\\cline{2-4}
        &AAMAS 2015  & 201     & 613        \\
        &AAMAS 2016 & 161      & 442        \\
        &Synthetic  & 100      & 5        \\
        \midrule
        \multirow{2}{*}{Sortition}&&  \multicolumn{2}{c}{\# Points}  \\\cline{2-4}
        &UCI-Adult& \multicolumn{2}{c}{17749}\\
        \bottomrule
    \end{tabular}
    \caption{Datasets}
    \label{tab:Datasets}
\end{table}

%% file: table/parameter.tex
\begin{table*}[ht!]
    \centering
    \begin{tabular}{lcccccc}
        \toprule
        \textbf{Dataset}&  &$N_{\eps}$& \multicolumn{2}{c}{$N_{rounds}$}   & \multicolumn{2}{c}{$N_{batches}$} \\
        \midrule               
        &&&$\smix$ &$\emix$&$\smix$ &$\emix$\\\cline{4-7}
        AAMAS 2015    &    & 50    & 100  & 50   & 10 & 5   \\
        AAMAS 2016    &   & 50   & 100   & 50   & 10  & 5  \\
        \multirow{3}{*}{Synthetic}    & $\epsilon=0.1$ & $\lceil\tfrac{2397}{1-\alpha}\rceil$  &100 &50    &10 & 5\\
            & $\epsilon=0.05$ &$\lceil\tfrac{11805}{1-\alpha}\rceil$   &100 &50    &10 & 5\\
            & $\epsilon=0.01$ &$\lceil\tfrac{423865}{1-\alpha}\rceil$  &100 &50    &10 & 5\\
        UCI-Adult    &  & 50    & 20   & 10    &5 &5    \\
        \bottomrule
    \end{tabular}
    \caption{Specifications for the experiments.}
    \label{tab:expsetup}
\end{table*}

%% file: tikz/sortition_center.tex
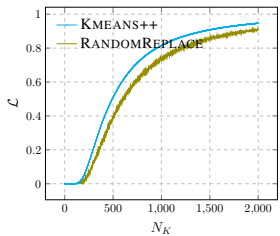
\begin{figure}[!h]
\hspace{18mm}
\begin{tikzpicture}[scale=0.45]
\begin{axis}[
    xlabel = {$N_K$},
    ylabel = {$\mathcal{L}$},
    table/col sep=comma,
    legend cell align = left,
    legend pos = north west,
    grid = major,
    grid style=dashed,
    legend style={nodes={scale=1.3, transform shape},fill=none,draw =none},
    legend columns=1,
    xlabel style={font=\large},
    ylabel style={font=\large},
    legend cell align={left},
] 
\addplot[     line width=0.2mm, 
    color=cyan,
    mark=|,
    mark size=0.5pt
    ] table[x=k,y=kmean++]{tikz/plot/adult_range_1_2000_clusters.csv};
\addplot[     line width=0.2mm, 
    color=olive,
    ] table[x=k,y=random]{tikz/plot/adult_range_1_2000_clusters.csv};
\legend{$\km$, $\rrq$}
\end{axis}
\end{tikzpicture}
    \caption{Number of the center $K$ affects Likelihood $\mathcal{L}$}
    \label{fig:cluster}
\end{figure}

%% file: arXiv.bbl
\newcommand{\etalchar}[1]{$^{#1}$}
\begin{thebibliography}{AAGW15}

\bibitem[AAB{\etalchar{+}}22]{survey2022}
Georgios Amanatidis, Haris Aziz, Georgios Birmpas, Aris Filos-Ratsikas, Bo~Li,
  Herv{\'e} Moulin, Alexandros~A Voudouris, and Xiaowei Wu.
\newblock Fair division of indivisible goods: {A} survey.
\newblock {\em arXiv preprint, arXiv:2208.08782}, 2022.

\bibitem[AAGW15]{aleksandrov2015online}
Martin~Damyanov Aleksandrov, Haris Aziz, Serge Gaspers, and Toby Walsh.
\newblock Online fair division: Analysing a food bank problem.
\newblock In {\em Proceedings of the 24th International Joint Conference on
  Artificial Intelligence {(IJCAI)}}, 2015.

\bibitem[AB22]{aziz2022vigilant}
Haris Aziz and Florian Brandl.
\newblock The vigilant eating rule: A general approach for probabilistic
  economic design with constraints.
\newblock {\em Games and Economic Behavior}, 135:168--187, 2022.

\bibitem[ABKR19]{arunachaleswaran2019fair}
Eshwar~Ram Arunachaleswaran, Siddharth Barman, Rachitesh Kumar, and Nidhi
  Rathi.
\newblock Fair and efficient cake division with connected pieces.
\newblock In {\em Proceedings of the 15th International Conference on Web and
  Internet Economics {(WINE)}}, pages 57--70. Springer, 2019.

\bibitem[AMS20]{aziz2020polynomial}
Haris Aziz, Herv{\'e} Moulin, and Fedor Sandomirskiy.
\newblock A polynomial-time algorithm for computing a {P}areto optimal and
  almost proportional allocation.
\newblock {\em Operations Research Letters}, 48(5):573--578, 2020.

\bibitem[AS98]{abdulkadirouglu1998random}
Atila Abdulkadiro{\u{g}}lu and Tayfun S{\"o}nmez.
\newblock Random serial dictatorship and the core from random endowments in
  house allocation problems.
\newblock {\em Econometrica}, 66(3):689--701, 1998.

\bibitem[AV07]{Kmeanpp}
David Arthur and Sergei Vassilvitskii.
\newblock K-means++: The advantages of careful seeding.
\newblock In {\em Proceedings of the Eighteenth Annual ACM-SIAM Symposium on
  Discrete Algorithms}, SODA '07, page 1027–1035, USA, 2007. Society for
  Industrial and Applied Mathematics.

\bibitem[Azi19]{aziz2019probabilistic}
Haris Aziz.
\newblock A probabilistic approach to voting, allocation, matching, and
  coalition formation.
\newblock In {\em The Future of Economic Design}, pages 45--50. Springer, 2019.

\bibitem[Azi20]{aziz2020simultaneously}
Haris Aziz.
\newblock Simultaneously achieving ex-ante and ex-post fairness.
\newblock In {\em Proceedings of the 18th International Conference on Web and
  Internet Economics {(WINE)}}, pages 341--355. Springer, 2020.

\bibitem[BBS20]{barman2020optimal}
Siddharth Barman, Umang Bhaskar, and Nisarg Shah.
\newblock Optimal bounds on the price of fairness for indivisible goods.
\newblock In {\em International Conference on Web and Internet Economics},
  pages 356--369. Springer, 2020.

\bibitem[BCKM13]{budish2013designing}
Eric Budish, Yeon-Koo Che, Fuhito Kojima, and Paul Milgrom.
\newblock Designing random allocation mechanisms: Theory and applications.
\newblock {\em The American Economic Review}, 103(2):585--623, 2013.

\bibitem[BCM16]{bouveret2016fair}
Sylvain Bouveret, Yann Chevaleyre, and Nicolas Maudet.
\newblock Fair allocation of indivisible goods.
\newblock In {\em Handbook of Social Choice}. Cambridge University Press, 2016.

\bibitem[BETS22]{PrefLib_AAMAS}
Rafael Bordini, Edith~Elkind Elkind, John Thangarajah, and David Shield.
\newblock Preflib: {AAMAS} bidding data (00037).
\newblock \url{https://www.preflib.org/dataset/00037}, 2022.

\bibitem[BL16]{bouveret2016characterizing}
Sylvain Bouveret and Michel Lema{\^\i}tre.
\newblock Characterizing conflicts in fair division of indivisible goods using
  a scale of criteria.
\newblock In {\em Proceedings of the 15th Autonomous Agents and Multi-Agent
  Systems ({AAMAS})}, pages 259--290. Springer, 2016.

\bibitem[BM01]{bogomolnaia2001new}
Anna Bogomolnaia and Herv{\'e} Moulin.
\newblock A new solution to the random assignment problem.
\newblock {\em Journal of Economic theory}, 100(2):295--328, 2001.

\bibitem[BT96]{brams1996fair}
Steven~J Brams and Alan~D Taylor.
\newblock {\em Fair Division: From cake-cutting to dispute resolution}.
\newblock Cambridge University Press, 1996.

\bibitem[Bud11]{budish2011combinatorial}
Eric Budish.
\newblock The combinatorial assignment problem: Approximate competitive
  equilibrium from equal incomes.
\newblock {\em Journal of Political Economy}, 119(6):1061--1103, 2011.

\bibitem[CKK21]{caragiannis2021interim}
Ioannis Caragiannis, Panagiotis Kanellopoulos, and Maria Kyropoulou.
\newblock On interim envy-free allocation lotteries.
\newblock In {\em Proceedings of the 22nd ACM Conference on Economics and
  Computation {(EC)}}, pages 264--284, 2021.

\bibitem[CKM{\etalchar{+}}19]{caragiannis2019unreasonable}
Ioannis Caragiannis, David Kurokawa, Herv{\'e} Moulin, Ariel~D Procaccia,
  Nisarg Shah, and Junxing Wang.
\newblock The unreasonable fairness of maximum nash welfare.
\newblock {\em ACM Transactions on Economics and Computation}, 7(3):1--32,
  2019.

\bibitem[CS02]{chen2002improving}
Yan Chen and Tayfun S{\"o}nmez.
\newblock Improving efficiency of on-campus housing: An experimental study.
\newblock {\em American economic review}, 92(5):1669--1686, 2002.

\bibitem[DG17]{UCI_adult}
Dheeru Dua and Casey Graff.
\newblock {UCI} machine learning repository, 2017.

\bibitem[EKM{\etalchar{+}}22]{ebadiansortition}
Soroush Ebadian, Gregory Kehne, Evi Micha, Ariel~D. Procaccia, and Nisarg Shah.
\newblock Is sortition both representative and fair?
\newblock In Alice~H. Oh, Alekh Agarwal, Danielle Belgrave, and Kyunghyun Cho,
  editors, {\em Advances in Neural Information Processing Systems}, 2022.

\bibitem[FGG{\etalchar{+}}21]{flanigan2021fair}
Bailey Flanigan, Paul G{\"o}lz, Anupam Gupta, Brett Hennig, and Ariel~D
  Procaccia.
\newblock Fair algorithms for selecting citizens’ assemblies.
\newblock {\em Nature}, 596(7873):548--552, 2021.

\bibitem[FSV20]{freeman2020best}
Rupert Freeman, Nisarg Shah, and Rohit Vaish.
\newblock Best of both worlds: Ex-ante and ex-post fairness in resource
  allocation.
\newblock In {\em Proceedings of the 21st ACM Conference on Economics and
  Computation {(EC)}}, pages 21--22, 2020.

\bibitem[Gal86]{10.1145/6462.6502}
Zvi Galil.
\newblock Efficient algorithms for finding maximum matching in graphs.
\newblock {\em ACM Comput. Surv.}, 18(1):23–38, mar 1986.

\bibitem[HZ79]{hylland1979efficient}
Aanund Hylland and Richard Zeckhauser.
\newblock The efficient allocation of individuals to positions.
\newblock {\em Journal of Political economy}, 87(2):293--314, 1979.

\bibitem[Mou04]{moulin2004fair}
Herv{\'e} Moulin.
\newblock {\em Fair division and collective welfare}.
\newblock MIT press, 2004.

\bibitem[NRR13]{nguyen2013survey}
Trung~Thanh Nguyen, Magnus Roos, and J{\"o}rg Rothe.
\newblock A survey of approximability and inapproximability results for social
  welfare optimization in multiagent resource allocation.
\newblock {\em Annals of Mathematics and Artificial Intelligence},
  68(1):65--90, 2013.

\bibitem[PM16]{procaccia2015cake}
Ariel~D. Procaccia and Hervé Moulin.
\newblock {\em Cake Cutting Algorithms}, page 311–330.
\newblock Cambridge University Press, 2016.

\bibitem[RR10]{roos2010complexity}
Magnus Roos and J{\"o}rg Rothe.
\newblock Complexity of social welfare optimization in multiagent resource
  allocation.
\newblock In {\em Proceedings of the 9th International Conference on Autonomous
  Agents and Multiagent Systems {(AAMAS)}}, pages 641--648, 2010.

\bibitem[RW98]{robertson1998cake}
Jack Robertson and William Webb.
\newblock {\em Cake-cutting algorithms: Be fair if you can}.
\newblock AK Peters/CRC Press, 1998.

\bibitem[SK02]{shavell2002fairness}
Steven Shavell and Louis Kaplow.
\newblock {\em Fairness versus Welfare}.
\newblock Harvard university press, 2002.

\bibitem[Str80]{stromquist1980cut}
Walter Stromquist.
\newblock How to cut a cake fairly.
\newblock {\em The American Mathematical Monthly}, 87(8):640--644, 1980.

\bibitem[Str08]{stromquist2008envy}
Walter Stromquist.
\newblock Envy-free cake divisions cannot be found by finite protocols.
\newblock {\em The electronic journal of combinatorics}, 15(1):R11, 2008.

\end{thebibliography}
